\documentclass[12pt]{article}

\usepackage{amssymb}
\usepackage{amsmath}
\usepackage{amsthm}
\usepackage{mathrsfs}
\usepackage{MnSymbol}
\usepackage{cite}
\usepackage{graphicx}
\usepackage{appendix}
\usepackage{enumerate}

\usepackage{calligra}
\DeclareMathAlphabet{\mathcalligra}{T1}{calligra}{m}{n}
\DeclareFontShape{T1}{calligra}{m}{n}{<->s*[2.2]callig15}{}

\usepackage[cal=boondox]{mathalfa}

\newcommand{\R}{\ensuremath{\mathbb{R}}}
\newcommand{\RP}{\ensuremath{\mathbb{R}_+}}
\newcommand{\scrp}{\ensuremath{\mathcal{p}}}
\newcommand{\scrq}{\ensuremath{\mathcal{q}}}
\newcommand{\scrm}{\ensuremath{\mathcal{m}}}

\newcommand{\scrv}{\ensuremath{\mathcal{v}}}
\newcommand{\scrw}{\ensuremath{\mathcal{w}}}

\newcommand{\scrP}{\ensuremath{\mathscr{P}}}

\newcommand{\scrB}{\ensuremath{\mathscr{B}}}
\newcommand{\scrO}{\ensuremath{\mathscr{O}}}

\newcommand{\scrV}{\ensuremath{\mathscr{V}}}
\newcommand{\scrM}{\ensuremath{\mathscr{M}}}
\newcommand{\scrI}{\ensuremath{\mathscr{I}}}

\newcommand{\Cone}{\ensuremath{\mathrm{Cone} \,}}

\newcommand{\coneP}{\ensuremath{\textnormal{Cone\,(\scrP)}}}
\newcommand{\theory}{\ensuremath{(\Sigma,\scrP)}}
\newcommand{\deltam}{\ensuremath{\Delta\scrm}}
\newcommand{\process}{\ensuremath{(\Delta\scrm,\scrq)}}
\newcommand{\MSigma}{\ensuremath{\scrM(\Sigma)}}
\newcommand{\MSigmaZ}{\ensuremath{\scrM^{\circ}(\Sigma)}}
\newcommand{\MSigmaPl}{\ensuremath{\scrM_{+}(\Sigma)}}
\newcommand{\MSigmaPlOne}{\ensuremath{\scrM_{+}^1(\Sigma)}}
\newcommand{\VSigma}{\ensuremath{\scrV(\Sigma)}}

\theoremstyle{plain}
\newtheorem{theorem}{Theorem}[section]

\newtheorem{lemma}[theorem]{Lemma}
\newtheorem{proposition}[theorem]{Proposition}

\theoremstyle{definition}
\newtheorem{definition}[theorem]{Definition}
\newtheorem{example}[theorem]{Example}

\theoremstyle{remark}
\newtheorem{rem}[theorem]{Remark}

\newtheorem*{rationale}{Rationale}
\newcounter{property}

\newenvironment{property}[1][]{\refstepcounter{property}\par\medskip\noindent%
    \textbf{Property~\theproperty. #1} \rmfamily}{\medskip}

\begin{document}

\author{Martin Feinberg\thanks{The William G. Lowrie Department of Chemical \& Biomolecular Engineering and Department of Mathematics, The Ohio State University, 151 W. Woodruff Avenue, Columbus, OH 43210 USA.  E-mail: feinberg.14@osu.edu.}
\and Richard B. Lavine\thanks{Department of Mathematics, University of Rochester, Rochester, NY 14627 USA.  Email: rdlavine@frontiernet.net}}

\title{Entropy and Thermodynamic Temperature in Nonequilibrium Classical Thermodynamics as Immediate Consequences of the Hahn-Banach Theorem: I. Existence}

\date{\emph{\today}}
\maketitle
\begin{abstract} The Kelvin-Planck statement of the Second Law of Thermodynamics is a stricture on the nature of heat receipt by any body suffering a cyclic process. It makes no mention of temperature or of entropy.  Beginning with a Kelvin-Planck statement of the Second Law, we show that entropy and temperature---in particular, existence of functions that relate the local specific entropy and thermodynamic temperature to the local state in a material body---emerge immediately and simultaneously as consequences of the Hahn-Banach Theorem. Existence of such functions of state  requires no stipulation that their domains be restricted to equilibrium states. Further properties, including uniqueness, are  addressed in a companion paper.
\end{abstract}
\newpage
\tableofcontents
\newpage
\numberwithin{equation}{section}

\section{Introduction}
\label{sec:Intro}

There are several widely-accepted formulations of the Second Law in classical thermodynamics, some invoking notions of temperature and entropy, and at least one invoking neither of these explicitly. 

In particular, the so-called Kelvin-Planck Second Law is an elemental stricture on the nature of heat receipt by a body during the course of any cyclic process the body might experience. It says, in effect, that during the course of a cyclic process the body cannot merely receive heat from its exterior without also emitting heat to it (in a manner qualitatively distinguishable from that of the heat receipt\footnote{We shall be more precise about this later on. Planck \cite{planck_treatise} requires that heat exchange between the body and its exterior cannot, on the whole, simply amount to \emph{extraction} of heat from a single ``heat reservoir." }). The First Law then implies that, over the course of a cyclic process, the heat received by the body cannot be converted entirely into work; there must be some heat emission as well. In the Kelvin-Planck Second Law there is no explicit mention of entropy or of temperature, much less of a thermodynamic temperature scale.

As we indicated, other invocations of the Second Law are explicit in their use of a thermodynamic temperature scale and of an entropy. Indeed the opening paragraph of Gibbs's ``On the Equilibrium of Heterogeneous Substances" \cite{gibbs_ehs} invokes an inequality of the form
\begin{equation}
\label{eq:CDInGibbs}
\begin{bmatrix}
\textrm{The total entropy}\\  
\textrm{of the body at the}\\  
\textrm{end of the process}
\end{bmatrix}
\,-\,
\begin{bmatrix}
\textrm{The total entropy of the}\\   
\textrm{body at the beginning}\\  
 \textrm{of the process}
\end{bmatrix}
\quad\geq\quad
{\int\frac{dq}{T}}\ \biggr\rvert_{\;\textrm{process}}
\end{equation}
\noindent
\emph{``dq denoting the element of heat received from external sources and $T$ denoting the temperature of the part of the system receiving it."} (This interpretation of the right side of \eqref{eq:CDInGibbs} is taken from that same Gibbs paragraph.)

Much of modern classical thermodynamics takes as its starting point a Second Law of the form \eqref{eq:CDInGibbs}, usually called the \emph{Clausius-Duhem inequality}, deemed to obtain for any body suffering any process, \emph{even processes in which there is rapid heating or cooling, in which there are sharp temperature gradients, and in which there is rapid and severe deformation}.\footnote{In the context of classical continuum  physics, a paper by Coleman and Noll \cite{coleman1963thermodynamics} greatly influenced modern methodology for use of the Clausius-Duhem inequality to deduce constraints on material properties. For  examples involving chemically reacting mixtures, see \cite{bowen1968thermochemistry} and Chapter 13 in \cite{feinberg2019foundations}.} Neither at the start of the process nor at its end need the body be in equilibrium.

This raises some historical and, more importantly, conceptual questions. Entropy and a thermodynamic temperature scale are generally regarded to be derived entities, deduced from more fundamental statements of the Second Law (such as the Kelvin-Planck version) by means of brilliant arguments posited by the early thermodynamics pioneers. Those arguments, however, often invoke idealized slow reversible processes (e.g., Carnot cycles) in which the body suffering the process is always in  (or arbitrarily close to) a condition of  equilibrium.\footnote{Fermi and Feynman give renditions of the classical arguments in \cite{fermithermo} and \cite{feynman1965}.}  Because the classically derived notions of entropy and thermodynamic temperature rest upon arguments in which the only body-states visited are ones at or very close to equilibrium, it is reasonable to question whether these notions actually have rigorous logical extensions to non-equilibrium domains. 

	Gibbs seemed willing to embrace such extensions. A reading of Gibbs's interpretation of the right side of \eqref{eq:CDInGibbs} indicates that he had no reluctance to invoke a thermodynamic temperature scale in bodies having different \emph{local} temperatures in different parts. And in an earlier, less read article \cite{gibbs_geom}, Gibbs clearly felt free to attribute an entropy to a body that is not in equilibrium:
	
\begin{quotation}
When the body is not in a state of thermodynamic equilibrium, its state is not one of those which are represented by our surface. The body, however, as a whole has a certain volume, entropy, and energy, which are equal to the sums of the volumes, etc., of its parts. 
\end{quotation}
\noindent
Note in particular that Gibbs was not reluctant to assert the existence of a \emph{local} entropy within an un-equilibrated  body, its total entropy coming from a summing process.

	Yet it is not easy to trace a clear path from the equilibrium arguments for entropy and thermodynamic temperature posited by the early pioneers to the non-equilibrium entropy and temperature invoked by Gibbs. Even less evident is a precise line of argument that begins with the pioneers and terminates with the free-wheeling modern use of local entropy and thermodynamic temperature in the Clausius-Duhem inequality, in particular when it is applied to bodies experiencing rapid, non-uniform heat transfer and deformation.
	
	Our aim is to connect, in a precise way, an elemental Kelvin-Planck statement of the Second Law to  the existence and properties of a thermodynamic temperature scale and an entropy scale, both viewed as functions of the local material state, that together satisfy the requirements of the Clausius-Duhem inequality (as it is invoked in modern classical physics\footnote{When we refer to classical or continuum physics, we mean that part of physics that embraces subjects such as fluid mechanics, heat transfer, elasticity theory, amd the theory of diffusive and reacting mixtures, in which bodies are regarded as continuous media.}) for all processes that material bodies under consideration are deemed to admit. The mathematical ideas we use, principally from functional analysis, were not available to the earliest pioneers of classical thermodynamics, nor were they available to Gibbs.\footnote{Henri Lebesgue was born in the year that \emph{Equilibrium of Heterogeneous Substances} \cite{gibbs_ehs} was first published.}	
	Our primary working tool is the Hahn-Banach Theorem, in particular a version that ensures that two non-empty disjoint closed convex sets in a locally convex topological vector space, at least one of them compact, can be strictly separated by a hyperplane \cite{brezis2011functional,	choquet1969lectures,rudin_functional_1991,
simon2011convexity}. Along the way, the Hahn-Banach Theorem will have the additional benefit of imparting to thermodynamics an intuitive geometric flavor, different in substance and setting from the geometric one pioneered by Gibbs in \cite{gibbs_graphical,gibbs_geom}.
	
\section{Some Background}This article and its companion \cite{feinberg-lavineEntropy2} constitute a major amplification of two much earlier ones by us, both drawing on the Hahn-Banach Theorem heavily. The first \cite{feinberg1983thermodynamics}, published in 1983,  was an extensive discussion of how the Hahn-Banach Theorem serves to connect a suitably formulated version of the Kelvin-Planck Second Law to the existence and properties (including uniqueness) of a thermodynamic temperature scale that conforms to the so-called \emph{Clausius inequality}---that is, to the Clausius-Duhem inequality \emph{restricted to cyclic processes}. For cyclic processes, the left side of \eqref{eq:CDInGibbs} reduces to  zero, so there was no involvement of entropy.

	The second article \cite{feinberg1986foundations} was published originally in 1984 as an appendix in \cite{truesdell_rational_1984} and soon after in a collection \cite{serrin_newperspectives} of short essays about the foundations of thermodynamics. That article indicated how, beginning with a slightly stronger version of a Kelvin-Planck Second Law, the Hahn-Banach Theorem delivers \emph{simultaneously} both a thermodynamic temperature and an entropy satisfying the requirements of the Clausius-Duhem inequality, \emph{not restricted to cyclic processes}. 	
	
	Although \cite{feinberg1986foundations} contained a Hahn-Banach proof of the \emph{equivalence} of the Kelvin-Planck Second Law with the existence of thermodynamic-temperature and entropy functions of state suited to the Clausius-Duhem inequality, several other theorems (including two   about  uniqueness) were merely stated without proofs. Those proofs we said would be forthcoming in a fuller article.  Moreover, we promised a more compelling presentation of the existence argument, in which certain presumptions about the structure of the putative set of thermodynamical processes would be substantially weakened. This  article and its companion  are intended to fulfill those promises. The weakened, and more natural, assumptions about the structure of the process-set have required a deeper analysis, much of it deferred to the appendix of this article. 
\bigskip

\begin{rem}
The 1986  volume \cite{serrin_newperspectives}, in which \cite{feinberg1986foundations} appears, contains a wealth of chapters by different authors devoted to the study of the mathematical foundations of non-equilibrium classical thermodynamics. The same is true of \cite{truesdell_rational_1984}.  Even beyond those, there are many schools of thought about how classical thermodynamics might be  extended to non-equilibrium settings. These are surveyed amply and critically in the 2008 book by Lebon, Jou, and Casa-V\'{a}squez \cite{lebon2008understanding} (although some work contained in \cite{serrin_newperspectives, truesdell_rational_1984} and related articles  escaped the book's notice). Readers of this article might want to see very different work by Lieb and Yngvason, which in 1999 \cite{lieb1999physics} began as an exploration of the construction of classical entropy for bodies in equilibrium and  then turned in 2013  to questions about the extent to which the same could be done for un-equilibrated bodies \cite{lieb2013entropy}.  For a recent summary of some of their work see \cite{yngvason2022direct}. See also an article by Kammerlander and Renner  \cite{kammerlander2020tangible}.
\end{rem}

\begin{rem} 
To a great extent, discussions with James Serrin in the late 1970s and early 1980s, in particular his formulations of the Second Law in terms of a heat accumulation function,
 provided inspiration for our work (although not our reliance on the Hahn-Banach theorem). Serrin's views at the time are captured in \cite{serrin1975foundations,serrin1978concepts,serrin1979conceptual,serrin1986outline}.
 
\end{rem}

\begin{rem} As in our earlier articles, we want to call particular attention to work \cite{vsilhavy1980measures,silhavy1980measures} by Miroslav \v{S}ilhav\'{y},\footnote{See also \v{S}ilhav\'{y}'s book \cite{silhavy1997mechanics}.} who realized independently and at about the same time that Hahn-Banach separation theorems, taken with the Kelvin-Planck Second Law, might provide a basis for existence of a thermodynamic temperature scale consistent with the \emph{cyclic-process} Clausius inequality.  In \cite{silhavy1980measures} \v{S}ilhav\'{y} viewed the thermodynamic temperature scale to be a function having as its domain a pre-supposed empirical temperature scale. The most apt comparison  to our work is with some preliminary notes \cite{feinberg-lavine_PrelimNotes} we wrote in 1978 for James Serrin. There, we also viewed a Clausius-inequality temperature scale to be a function having as its domain a pre-supposed empirical temperature scale, and we too used Hahn-Banach separation theorem arguments to demonstrate how the existence of such a Clausius-inequality temperature scale derives immediately from, and is equivalent to, the Kelvin-Planck Second Law.

Our subsequent published article \cite{feinberg1983thermodynamics} on the Clausius inequality  was much more ambitious. There, we chose not to pre-suppose an empirical temperature scale, carrying a pre-ordained notion of ``hotness" and ``hotter than."  Rather, we regarded the desired Clausius-inequality temperature scale to be a ``function of state," the state domain depending on the material under consideration.\footnote{For example, the local state of a gas might be  specified by the local pressure and the local specific volume. For other examples see \S\ref{sec:StateSpaces}.}  In this way, we could not only establish, via the Hahn-Banach Theorem, the equivalence of the Kelvin-Planck Second Law with the existence of a temperature scale satisfying the Clausius inequality, we could also tie relative values of that temperature to a ``hotter than" relation on the set of states, a relation deriving  solely from processes the material is deemed to admit. This is the position taken in \cite{feinberg1986foundations} and here, where the entropy density, like the thermodynamic temperature scale, is a Hahn-Banach-derived function of the local material state.\footnote{After establishing a thermodynamic temperature function (of the empirical temperature), one suited to the cyclic-process Clausius inequality, \v{S}ilhav\'{y} \cite{silhavy1980measures}, went on to construct an entropy, but that entropy is an attribute of an entire body, not the entropy-density function of the local state established here via the Hahn-Banach Theorem.}
\end{rem}

\section{Thermodynamical Theories}

To a great extent modern classical thermodynamics manifests itself as a collection of thermodynamical theories tailored to particular materials, these various theories sharing common premises and common methodologies. There are, for example, thermodynamical theories of elastic materials, of gases, of viscous fluids, of diffusive reacting mixtures, and so on. Each such theory presumably carries with it versions of the First and Second Laws, rendered concrete and precise within the context of the specific class of materials under study. 

	With this viewpoint in mind, we regard the theorems contained in this article and its companion to provide something like a ``meta-thermodynamics" that sheds an overarching light on the structure of specific thermodynamical theories. In particular, almost all of the theorems contained here assert that a theory has Property A (usually a statement about the nature of heat transfer between bodies  and their exteriors in processes the theory admits) \emph{if and only if} it has Property B (usually a statement about entropy and thermodynamic temperature). The deeper and more difficult of those implications always derives from the Hahn-Banach Theorem.
	
	We will regard a thermodynamical theory to be a mathematical object consisting of two sets: (i) a \emph{state space} $\Sigma$ that characterizes the set of (local) states that might be exhibited within a material body embraced by the theory and (ii) a set \scrP\, of \emph{processes} that  abstracts the essential features of physical processes that such bodies are deemed to admit.   Taken together, these two sets will, for us, serve to constitute an instance $(\Sigma,\scrP)$ of a  \emph{thermodynamical theory}.
	
	In this section and the next we will use terms such as \emph{body}, \emph{material}, \emph{material point}, and \emph{physical process}, but only in an informal way to guide thinking about the two sets $\Sigma$ and \scrP\  that constitute a thermodynamical theory or to provide justification for the structure these sets are presumed to possess. Again, though, a thermodynamical theory $(\Sigma,\scrP)$ is a purely \emph{mathematical} object suited to precisely stated   questions and theorems.  In particular, we will be in a position to say what we mean by a \emph{Kelvin-Planck theory}---that is, a thermodynamical theory that complies with a precisely stated version of the Kelvin-Planck Second Law. And we will be a position to ask about circumstances under which  a particular thermodynamical theory $(\Sigma,\scrP)$ admits two functions of state---a specific-entropy $\eta:\Sigma\to\R$ and a thermodynamic temperature scale $T:\Sigma\to\R_+$ that together comply with the Clausius-Duhem inequality for all processes \scrP\  the theory contains.

\begin{rem} The mathematical objects and theorems contained here lend themselves to a variety of physical interpretations. At least at the outset, it will be helpful for the reader to think of a  thermodynamic theory \theory\   as a description of a particular material (e.g., carbon dioxide, water, rubber, a metal alloy, a diffusive reacting mixture). In this context, a specific-entropy function $\eta:\Sigma\to\R$ will have an interpretation as an attribute of a particular material---in the parlance of continuum physics, a ``constitutive function" for that material. Nevertheless, we intend the abstract idea of a thermodynamic theory to be broadly adaptable to a variety of circumstances and instances. 
\end{rem}

\subsection{State Spaces}\label{sec:StateSpaces}

Central to virtually all classical theories of material body behavior is the idea of ``functions of state" that serve to compute \emph{local} values of certain material attributes. Indeed, one of our aims is to establish, from the Kelvin-Planck Second Law, the existence of specific-entropy and thermodynamic-temperature functions, suited to the Clausius-Duhem inequality, that permit the calculation of the local specific entropy (entropy  per mass) and the local thermodynamic temperature once the local material ``state" is specified.

Just how the ``state of a material point" is specified will vary from one thermodynamical theory to another.\footnote{The idea of a material point is basic to classical physics, wherein reference is freely made to the density, velocity, stress tensor, temperature, or species concentrations at a point within a body. We will always regard a \emph{state} as an attribute of a material point within a body, not of the body as a whole. For the body as a whole we will refer to its \emph{condition}.} For a theory of a gas of fixed composition it might be supposed that the local state is captured completely by specification of the pair $(p,v)$, where $p$ is the local pressure and $v$ is the local specific volume (the reciprocal of the density). For an elastic material, it might be supposed that the local state is captured by the pair $(u,F)$, where $u$ is the local specific internal energy (internal energy per mass) and $F$ is the local deformation gradient. For a reacting and diffusive mixture having $n$ chemical species, the local state might be described by the vector $[c_1,c_2,\dots,c_n,\theta] \in \mathbb{R}^{n+1}$, where $c_i$ is the local molar concentration of the $i^{th}$ species and $\theta$ is the local temperature in degrees Fahrenheit.

	In any case, we shall take for granted that a thermodynamical theory  has associated with it a \emph{state space} $\Sigma$, understood to be the set of local states that might be exhibited within a material body during processes the theory purports to describe. It will be presumed that $\Sigma$ carries with it a Hausdorff topology.

In fact, we will go further by supposing hereafter that $\Sigma$ is compact. This supposition will simplify the mathematics greatly, and in most instances it will be physically apt:  A well-grounded theory would suffer no loss from exclusion of processes that visit material states which are physically unreasonable. Excluded from consideration, for example, might be processes involving mass-densities so high as to be realized only in black holes or so low as to be inconsistent with the tenets of continuum models.

\begin{rem}\label{rem:LocallyCompact} When the state space is merely presumed to be locally compact, realization of the  objectives of this paper become more technically delicate, and certain theorems here become false without modification.  In Appendix E of \cite{feinberg1983thermodynamics} we showed how this might proceed when attention is restricted to solely to cyclic processes, with the aim of producing a thermodynamic temperature scale consistent with the Clausius inequality.
\end{rem}

\subsection{Processes}

	A process experienced by particular body can be described in a variety of ways, some highly picturesque, involving pulleys and pistons. For our purposes, however, there will be only two aspects of the process that need be considered: (i) the \emph{change of condition} of the body from the beginning of the process to its end and (ii) the \emph{heating measure} for the process, which is an overall accounting of the nature of heat receipt the body experiences during the course of the process. We will describe each of these separately. For us, a \emph{process} will be identified with specifications of both its change of condition and its heating measure.

\subsubsection{The change of condition for a process} Recall that members of $\Sigma$ are understood to be \emph{local} state descriptions---that is, candidates for describing the state of a material point within a body. If we consider a body at a fixed instant, its material points will be exhibited in various states of $\Sigma$. Although there might be just one state exhibited throughout the body (in which case the body is thermodynamically uniform), the distribution of states over the body could be far more diffuse. In any case, we shall need a device to describe that distribution for a particular body at a fixed instant.

By the (instantaneous) \emph{condition} of the body we mean a positive regular Borel measure on $\Sigma$, denoted here by  \scrm, interpreted in the following way: For each Borel set $\Lambda \subset \Sigma$, $\scrm(\Lambda)$ is the mass of that part of the body consisting of all material points in states contained in $\Lambda$. More colloquially, we can think of $\scrm(\Lambda)$ to be determined by excising from the body only material in states contained within $\Lambda$ and weighing that part of the  body so removed. Note that $\scrm(\Sigma)$ is the mass of the entire body. Note also that if a body of mass $M$ is thermodynamically uniform, with all material in state $\sigma$, then the body's condition is $M\delta_{\sigma}$, where $\delta_{\sigma}$ is the Dirac measure concentrated at $\sigma$.\footnote{The Dirac measure $\delta_{\sigma}$ is defined by the requirement that, for each Borel set $\Lambda \subset \Sigma$, $\delta_{\sigma}(\Lambda)$ is either $1$ or $0$ according to whether $\sigma$ is or is not a member of $\Lambda$.}

	Now consider a physical process suffered by a particular body, with both the body and the process presumably embraced by the thermodynamic theory under consideration. During the process, the body might experience rapid deformation and heat treanser, so that each material point within the body might present itself in a great variety of states as the process ensues. In particular, the body's final condition $\scrm_f$  might be very different from the body's initial condition $\scrm_i$. We associate with the process a \emph{change of condition}, $\Delta\scrm$ defined by
\begin{equation}
\Delta\scrm := \scrm_f - \scrm_i.
\end{equation}

Here $\Delta\scrm$ is understood to be a signed regular Borel measure on $\Sigma$, which is to say that $\Delta\scrm$ might take positive values on some Borel sets and negative values on others.\footnote{When the compact Hausdorff topology of $\Sigma$ is given by a metric, in particular in the almost universal case in which the state space is taken to be a subset of $\mathbb{R}^N$, every finite signed Borel measure is  already regular. See Chapter 12 in \cite{charalambos2013infinite}.} Note, however, that we always have
\begin{equation}
\label{eq:deltamSigIsZero}
\Delta\scrm(\Sigma) = \scrm_f(\Sigma) - \scrm_i(\Sigma) = 0,
\end{equation}
since each term on the right is the (conserved) total mass of the body suffering the process.

\subsubsection{The heating measure for a process} During the course of the physical process under consideration, the body suffering the process might experience deformation and nonuniform transfer of heat to and from its exterior. Indeed, at a given instant there might be heat receipt in some parts of the body and heat removal in other parts. It should be kept in mind that each material point can be expected to visit a variety of states in $\Sigma$ as time progresses.

	With the process we associate a \emph{heating measure} \scrq, which is a signed regular Borel measure on $\Sigma$ with the following interpretation: For each Borel set $\Lambda \subset \Sigma$, $\scrq(\Lambda)$ is the net amount of heat received over the course of the entire process (from the exterior of the body suffering the process) by material in states contained within $\Lambda$  at the time of heat receipt. In colloquial terms, imagine viewing the evolving process through glasses that filter out material not in states contained in $\Lambda$; some material might disappear and then reappear. The net heat received, over the entire process, by the visible material (from the exterior of the entire body) is  $\scrq(\Lambda)$.
	
\subsubsection[A robust example]{Example: Change of condition and heating measure derived from a more concrete process description}\label{subsubsec:ContinuumMechEx} Because the abstract idea of a process's change of condition and heating measure will be important hereafter,\footnote{The idea of expressing the condition of a body as a measure on a state space was inspired by a paper by Noll \cite{noll1970certain}. As far as we can recall from private conversations in 1978, James Serrin had  invented what we call a heating measure, but on a one-dimensional ``hotness manifold." He later abandoned that in published works, as he came to favor what he called a heat ``accumulation function" on the hotness manifold \cite{serrin1979conceptual,serrin1986outline}. We do not take hotness as a primitive notion.} we will indicate how these can be calculated from a somewhat more tangible description of a process. With the process (having a compact metric space as the state space $\Sigma$) we associate:

\begin{enumerate}[(i)]

\item a body \scrB\  that experiences the process. Here we regard \scrB\  to be a set (of \emph{material points}), taken with a $\sigma$-algebra of subsets of \scrB, called the \emph{parts} of \scrB.  We presume that \scrB\  comes equipped with a positive \emph{mass measure} $\mu$ defined on its parts: for each part $P \in \scrB$, $\mu(P)$ is the \emph{mass} of part $P$. 

\item a closed interval of the real line  $\scrI := [t_i,t_f]$, identified with the time interval over which the process transpires.

\item a measurable\footnote{Here $\Sigma$ is understood to carry the Borel $\sigma$-algebra.}  function $\hat{\sigma}: \scrB \times \scrI \to \Sigma$, with $\hat{\sigma}(X,t)$ interpreted as the state of material point $X$ at instant $t$.

\item a real-valued signed measure $h$  on $\scrB \times \scrI$, interpreted as follows: For each part $P \subset \scrB$ and each Lebesgue-measurable set $J \subset \scrI$, $h(P,J)$ is the net amount of heat received by part $P$ from the exterior of the body during instants contained in $J$.
\end{enumerate}

	For a process described this way, construction of the \emph{heating measure} \scrq\  proceeds as follows: For each Borel set $\Lambda \subset \Sigma$,
\begin{equation}
\scrq(\Lambda) := h(\hat{\sigma}^{-1}(\Lambda)).
\end{equation}
To construct the change of condition for the process we begin by defining the initial and final state assignments to material points:
\begin{equation}
\hat{\sigma}_i(\cdot) := \hat{\sigma}(\cdot,t_i) \quad \text{and} \quad \hat{\sigma}_f(\cdot) := \hat{\sigma}(\cdot,t_f).
\end{equation}
The \emph{initial condition} and \emph{final condition} of body \scrB\  are then defined by the requirement that, for each Borel set $\Lambda \subset \Sigma$,
\begin{equation}
\scrm_i(\Lambda) = \mu(\hat{\sigma}_i^{-1}(\Lambda)) \quad \text{and}\quad \scrm_f(\Lambda) = \mu(\hat{\sigma}_f^{-1}(\Lambda)).
\end{equation} 
The \emph{change of condition} for the process is then  given by
\begin{equation}
\Delta\scrm := \scrm_f - \scrm_i.
\end{equation}
%Because $\Sigma$ is a compact metric space, $\Delta\scrm$ and \scrq\ are regular \cite{charalambos2013infinite}.

\subsubsection{The set of processes and some of its properties}\label{subsubsec:ProcProps} In a theory with state space $\Sigma$, a process will be regarded to be a pair \process, where $\Delta \scrm$ is the change of condition for the process and \scrq\  is its heating measure. We can regard both of these as members of $\scrM(\Sigma)$, the vector space of signed regular Borel measures on $\Sigma$. In fact, from  \eqref{eq:deltamSigIsZero} it follows that \deltam\  is always a member of the linear subspace  $\MSigmaZ \subset \MSigma$ defined by
\begin{equation}
\MSigmaZ := \{\nu \in \MSigma\  :\  \nu(\Sigma) = 0\}.
\end{equation}
Thus we can regard a process \scrp\, = \process\  to be a member of the vector space
\begin{equation}
\scrV(\Sigma) := \MSigmaZ \oplus \MSigma. 
\end{equation}
Hereafter it will be understood that \MSigma\  carries the weak-star topology,\footnote{\label{ft:WeakStarDef}The weak-star topology on \MSigma\  is its coarsest topology such that, for every continuous function $\varphi:\Sigma \to \mathbb{R}$, the map 
\begin{equation}
\scrv \in \MSigma \to \int_\Sigma \varphi\, d \scrv \nonumber
\end{equation} 
is continuous. Then $\scrv_0 \in \MSigma$ is in the weak-star closure of a subset $S \subset \MSigma$ if, for every finite sequence $\varphi_1, \varphi_2,\dots,\varphi_n$ in $\mathrm{C}(\Sigma,\mathbb{R})$ and every $\varepsilon > 0$, there is a $\scrv \in S$ such that $|\int_{\Sigma}\varphi_j d \scrv - \int_{\Sigma}\varphi_j d \scrv_0\,| < \varepsilon,\ j=1,\dots,n$. Unlike the norm topology, the weak-star topology on \MSigma\ reflects the topology of $\Sigma$. For example, if $\sigma_i \to \sigma$ in $\Sigma$  as $i \to \infty$, then $\delta_{\sigma_i} \to \delta_{\sigma}$ with respect to the weak-star topology in \MSigma, while $\|\delta_{\sigma_i} - \delta_{\sigma}\| =2$ for each $i$.}
 that \MSigmaZ\  carries the topology it inherits as a subset of \MSigma, and that \VSigma\  carries the resulting product topology. For a set $X \in \VSigma$ we denote by cl\,$(X)$ its closure.
 
 	For a thermodynamic theory with state space $\Sigma$, the \emph{set of processes}, \mbox{$\scrP\  \subset  \VSigma$}, will be understood to consist of members of \VSigma\  that correspond to physical processes deemed to be admitted by material bodies in circumstances the theory purports to embrace.  Physical considerations suggest that, for any reasonable theory, the set \scrP\  should carry a certain structure, in particular that it should share at least some of the attributes of a convex cone in \VSigma. Recall that \scrP\  would be a convex cone were it to have both of the following properties:
\begin{enumerate}[(i)]
\item For each \scrp\  in \scrP\  and each non-negative number $\alpha$, $\alpha\scrp$ is a member of \scrP. 
\item For all \scrp\ and  $\scrp^*$ in \scrP, $\scrp + \scrp^*$ is a member of \scrP.
\end{enumerate}

	With respect to (i), it is not difficult to argue on physical grounds that that the inclusion will be satisfied \emph{so long as $\alpha$ is a non-negative integer}: If \scrp = \process\  is a physical process suffered by a body \scrB, then for any positive integer $n$, we can simultaneously execute the same process on $n$ copies of \scrB, copies that are not in thermal communication. The $n$ bodies, viewed as a single body, will have suffered a physical process for which the change of condition is $n\deltam$ and the heating measure is $n\scrq$. Thus, $n\scrp = (n\deltam, n\scrq)$ is a member of \scrP, corresponding to the physical $n$-body process described.
	
	Similarly, we can expect on physical grounds that the inclusion in (ii) will be satisfied so long as \scrp\  = \process\  and $\scrp^* =(\deltam^*,\scrq^*)$ correspond to two physical processes \emph{having the same temporal duration}: If these physical processes are suffered by bodies \scrB\  and $\scrB^*$, then the two processes can be executed simultaneously, with \scrB\  and $\scrB^*$ thermally isolated from one another, perhaps by large physical distance. This simultaneous execution can be viewed to be another physical process, suffered by the body composed of  \scrB\  and $\scrB^*$, having change of condition \mbox{$\deltam + \deltam^*$} and heating measure $\scrq + \scrq*$. In this case, the new physical process would have a representation in \VSigma\ (and  in \scrP) given by $\scrp + \scrp^*$. 
	
\medskip

	These considerations tell us that, in a reasonable theory, the process set \scrP\ can be expected to have some natural structure, including features that are \emph{suggestive} of a convex cone in \VSigma. In fact, in \cite{feinberg1986foundations} we \emph{assumed} that \scrP\ is a convex cone. Here we make no such assumption.
	
	We defer to the Appendix  a far more nuanced discussion of the structure that we will suppose \scrP\  possesses. By $\Cone(\scrP)$ we mean the set in \VSigma\ defined by
\begin{equation}
\Cone(\scrP) := \{\alpha\scrp \in \VSigma:\  \scrp \in \scrP,  \alpha \geq 0\}. 
\end{equation}
Based on a few plausible physical assumptions, we argue in the Appendix that, in a reasonable theory, the set
\begin{equation}
%\label{eq:defCHat}
\hat{\scrP} := \textrm{cl}\,(\Cone(\scrP))
\end{equation}	
should not only be a closed cone in \VSigma, it should also be \emph{convex}.  This we will take for granted hereafter.

\subsection{Definition of a Thermodynamical Theory}
For the record, we posit the following definition:

\begin{samepage}
\begin{definition}\label{def:ThermTheory} A \textbf{thermodynamical theory} consists of a (compact) Hausdorff set $\Sigma$, called the \emph{state space} of the theory, and  a set $\scrP \subset \VSigma$ such that 
\begin{equation}
%\label{eq:defCHat}
\hat{\scrP} := \textrm{cl}\,(\Cone(\scrP))
\end{equation}
is convex. Elements of \scrP\  are the \emph{processes} of the theory.
\end{definition}
\end{samepage}

\begin{rem} The definition is formulated in such a way as to remind the reader of our presumption that $\Sigma$\ is compact. Recall Remark \ref{rem:LocallyCompact}. 
\end{rem}

\section{Kelvin-Planck Theories}

	In this section we will make precise what we mean by a \emph{Kelvin-Planck theory}---that is, a thermodynamical theory that respects a form of the Kelvin-Planck Second Law. We want to capture the following idea: In every cyclic process in which the body suffering the process experiences a  heat absorption from the body's exterior, there must also be heat emission to the exterior, the emission being qualitatively different from the absorption. If there were there no heat emission, the process would be perfectly efficient, for by the First Law the heat absorbed would be converted entirely into work. 
	
	By a cyclic process in the thermodynamical theory \theory\ we will mean a process in which the condition of the body at the end of the process is the same as it was at its beginning. That is, \emph{a cyclic process $\scrp = \process$ is a process such that the change of condition $\deltam$ is $0$}. 
	
	Consider a cyclic process $\scrp^* := (0,\scrq^*)$ with $\scrq^* \neq0$. Recall that if $\Lambda \subset \Sigma$ is a Borel set of states, then $\scrq^*(\Lambda)$ is interpreted to be the net amount of heat absorbed during the course of the entire process by material while in states contained in $\Lambda$. If $\scrq^*$ is a \emph{non-negative} Borel measure---that is, one that takes non-negative values on \emph{every} Borel set, then there is no Borel set of states that, for the process, can be associated with net heat \emph{emission}. Moreover, by supposition  $\scrq^*$ is not the zero measure, so there is at least one Borel set on which  $\scrq^*$ is positive, corresponding to heat \emph{absorption}.

	For these reasons, \emph{when $\scrq^* \neq 0$ is a non-negative measure, we will regard the cyclic process $\scrp^* := (0,\scrq^*)$ to be inconsistent with the spirit of the Kelvin-Planck Second Law.} 	For the thermodynamical theory \theory\ we denote by $\scrM_+(\Sigma)$ the set of non-negative regular Borel measures on $\Sigma$, and we also let
\begin{equation}
(0,\scrM_+(\Sigma)) := \{ (0,\scrv) \in \VSigma : \scrv \in \scrM_+(\Sigma)\}. \nonumber
\end{equation}
Thus, for a thermodynamic theory \theory\ we might regard the requirement 
\begin{equation}
\label{eq:KPver1}
\scrP\, \cap\, (0,\scrM_+(\Sigma))\; \textrm{is at most}\; (0,0)
\end{equation}
to be a full embodiment of the Kelvin-Planck Second Law. Or, if we want to assert that a nonzero element of  $(0,\scrM_+(\Sigma))$   cannot even be approximated by the theory's processes, then we might strengthen \eqref{eq:KPver1} by requiring that
\begin{equation}
\label{eq:KPver2}
\textnormal{cl}\,(\scrP)\, \cap\, (0,\scrM_+(\Sigma))\; \textrm{is at most}\; (0,0).
\end{equation}

	However, two examples will reveal a sense in which even \eqref{eq:KPver2} falls a little short of capturing the Kelvin-Planck stricture against an approach to perfect conversion of heat into work in cyclic processes. The examples will indicate why we prefer to express the Kelvin-Planck Second Law in terms of a requirement that is somewhat stronger than \eqref{eq:KPver2}.

	Each example will be in the form of a toy thermodynamic theory in which the state space $\Sigma$ is identified with the real interval $[0,1]$. Recall that, for $x \in \Sigma$, $\delta_x$ denotes the Dirac measure at $x$. That is, if $\Lambda \subset \Sigma$ is a Borel set then $\delta_x(\Lambda) = 1$ if $x$ is in $\Lambda$ and is zero otherwise.

\begin{example}[\emph{A sequence of cyclic processes with small fixed heat emission but unbounded heat receipt}]
\label{ExampA}
 Consider a thermodynamic theory \theory, in which \scrP\  contains the sequence of cyclic processes 
\begin{equation}
\label{eq:BadProcessSeqA}
 \{(0,n\delta_1 - \delta_0)\,: n = 1,2,...\}.
\end{equation} 
Note that in each process of the sequence there is heat absorbed (by material in state $1$) and heat emitted (by material in state $0$). Thus, no process of the sequence is a member of the forbidden set $(0,\MSigmaPl)$, nor does the sequence converge to any nonzero member of the forbidden set. For this reason, a putative assertion of the Kelvin-Planck Second Law in the form \eqref{eq:KPver2} would not preclude for the theory \theory\  the presence of the sequence  \eqref{eq:BadProcessSeqA} in \scrP.

Nevertheless, the sequence contains cyclic processes \emph{that come arbitrarily close to having perfect efficiency as $n$ increases}: In each process, the heat \emph{absorbed} (all at state $1$) is $n$, while the work done (equal, in a cyclic process, to the \emph{net} amount of heat received) is $n-1$. The efficiency, then, is $\frac{n-1}{n}$, which approaches $1$ as $n$ gets large. Although members of the sequence \eqref{eq:BadProcessSeqA} do not converge to a member of the forbidden set, they do come close to \emph{aligning} in the vector space \VSigma\  with the forbidden element $(0,\delta_1)$.

	Such an arbitrarily close approach to perfect efficiency would seem to violate the spirit of the Kelvin-Planck Second Law. The example reveals a sense in which the condition expressed by \eqref{eq:KPver2} is not a fully suitable reflection of that spirit. 
\end{example}
	
\begin{example}[\emph{A sequence of almost-cyclic processes, each with heat receipt but no heat emission}]
\label{ExampB}
 Consider a thermodynamic theory \theory, in which \scrP\  contains the sequence of  processes 
\begin{equation}\label{eq:BadProcessSeq2}
\{(\delta_{1/n} - \delta_0,n\delta_{1/2})\,: n = 1,2,...\}.
\end{equation} 
Note that, in each process of the sequence, the heating measure indicates no heat emission, only (unbounded) heat absorption, entirely at state $\frac{1}{2}$. Still, no process of the sequence  constitutes a violation of a Kelvin-Planck-type Second Law, as no process is cyclic. Nevertheless, as  $n$ increases the change of condition approaches $0$ while the heat absorption becomes unbounded. Although the sequence does not converge to any member of the forbidden set $(0,\MSigmaPl)$, its processes nevertheless violate the Kelvin-Planck spirit, for as $n$ increases they increasingly resemble cyclic processes with (large) heat absorption but no heat emission. 

	Here, as in Example \ref{ExampA}, a codification of the Kelvin-Planck Second Law in the form  \eqref{eq:KPver2}  does not suffice to preclude the presence in \scrP\  of a troubling process sequence, in this case  \eqref{eq:BadProcessSeq2}. 
\end{example}
\medskip	

	Stated informally, the difficulty in both examples is that, while neither sequence converges to an element of the forbidden set $(0,\scrM_+(\Sigma))$, members of each sequence come arbitrarily close to pointing along a ``forbidden direction"	in the vector space \VSigma. 
	
	For a thermodynamic theory \theory, we will identify the \emph{direction} of a process $\scrp \in \scrP$ with the half-line 
\begin{equation}
\{\alpha \scrp \in \VSigma : \alpha \geq 0\}.
\end{equation}
Note that $\Cone(\scrP)$, given as before by
\begin{equation}
\Cone(\scrP) := \{\alpha\scrp \in \VSigma :\  \scrp \in \scrP,  \alpha \geq 0\}, 
\end{equation}
is the set of all directions generated by members of \scrP. The condition
\begin{equation}
\label{eq:KPver3}
\textnormal{cl}\,(\Cone(\scrP))\ \cap\ (0,\MSigmaPl) = (0,0)
\end{equation} 
then says in effect that no nonzero element of the forbidden set (0,\MSigmaPl) can be approximated by vectors  of \VSigma\ having directions associated with members of \scrP. 
\medskip

\begin{rem}[\emph{Examples \ref{ExampA} and \ref{ExampB} reconsidered}] Although the problematic thermodynamical theories considered in Examples \ref{ExampA} and \ref{ExampB} were not precluded by the putative Kelvin-Planck Second Law in the form \eqref{eq:KPver2}, they are precluded by the strengthened condition \eqref{eq:KPver3}. In the case of Example \ref{ExampA} the sequence in Cone\,(\scrP)
\begin{equation}\nonumber
 \{(0,\delta_1 - \frac{1}{n}\delta_0)\,: n = 1,2,...\}\end{equation} 
 converges to $(0,\delta_1)$. In the case of Example \ref{ExampB} the sequence in Cone\,(\scrP)
\begin{equation}\nonumber
 \{(\frac{1}{n}[\delta_{1/n} - \delta_0],\delta_{1/2})\,: n = 1,2,...\}
 \end{equation} 
converges to $(0,\delta_{1/2})$.
\end{rem}

\bigskip

	For these reasons, our preferred codification of the Kelvin-Planck Second Law will take the form \eqref{eq:KPver3} rather than \eqref{eq:KPver2}. \emph{Note that if \scrP\ is itself a cone then there is no difference between \eqref{eq:KPver3} and \eqref{eq:KPver2}.}	Recall that in Definition \ref{def:ThermTheory} (the definition of a thermodynamical theory \theory) we let
\begin{equation}
\label{eq:defCHat}
\hat{\scrP} := \textrm{cl}\,(\Cone(\scrP)).
\end{equation}

\begin{definition} A \textbf{Kelvin-Planck theory} is a thermodynamical theory \theory\ such that
\begin{equation}
\label{eq:KPSecLawDef}
\hat{\scrP}\ \cap\ (0,\MSigmaPl) = (0,0).
\end{equation}
\end{definition}

\section[The Hahn-Banach Theorem and the Second Law]{Hahn-Banach Equivalence of the Kelvin-Planck Second Law and the Existence of Entropy-Temperature Functions of State} 

The following theorem asserts that, for a thermodynamic theory, compliance with the Kelvin-Planck Second Law is \emph{equivalent} to the existence of two continuous functions of state, a specific-entropy function and a thermodynamic temperature scale that, taken together, satisfy the Clausius-Duhem inequality for all processes the theory contains. \emph{Entropy and thermodynamic temperature emerge simultaneously and almost immediately as a direct consequence of  the Hahn-Banach theorem.  There is no reliance at all on venerable thermodynamic conceptual machinery in the form of reversible processes, Carnot cycles,  heat baths, or even the idea of equilibrium.} 

	In the theorem statement $\textnormal{C}(\Sigma,\mathbb{R})$ denotes the set of real-valued continuous functions on $\Sigma$, and $\textnormal{C}(\Sigma,\mathbb{R}_+)$ is the set of positive-valued continuous functions. $\mathbb{R}_+$ denotes the set of strictly positive real numbers. 
\begin{samepage}
\begin{theorem}[Existence of Entropy and Thermodynamic Temperature]
\label{thm:ExistTempEnt}
 For a thermodynamic theory \theory\ the following are equivalent:
\begin{enumerate}[(i)]
\item \theory\ is a Kelvin-Planck theory.
\item There exist functions $\eta \in \textnormal{C}(\Sigma,\mathbb{R})$ and $T  \in \textnormal{C}(\Sigma,\mathbb{R}_+)$  such that
\begin{equation}\label{eq:MainThmCDIneq}
\int_{\Sigma}\eta\: d(\deltam)\  \geq\  \int_{\Sigma}\frac{d\scrq}{T}, \quad \forall\  (\deltam,\scrq) \in \scrP.
\end{equation}
\end{enumerate}
\end{theorem}
\end{samepage}
\bigskip

\bigskip
	Proof of Theorem \ref{thm:ExistTempEnt} will make use of some fairly straightforward adaptations of ideas (see, for example, \cite{choquet1969lectures}) in functional analysis  that were unavailable to the thermodynamics pioneers: First, \VSigma\  is a locally convex Hausdorff topological vector space. Second, the compactness of $\Sigma$ ensures that the convex set 
\begin{equation}
(0,\MSigmaPlOne) := \{(0,\scrv) \in \VSigma : \scrv \in \MSigmaPl, \scrv(\Sigma) = 1\} 
\end{equation}
is (weak-star) compact. Finally, if $f:\VSigma \to \mathbb{R}$ is a continuous linear function, then there exist functions $\alpha(\cdot)$ and $\beta(\cdot)$ in $\textnormal{C}(\Sigma,\mathbb{R})$ such that, for every $(\scrv,\scrw) \in \VSigma$,
\begin{equation}\label{eq:RepThm}
f(\scrv,\scrw) = \int_{\Sigma}\alpha\, d\scrv +  \int_{\Sigma}\beta\, d\scrw.
\end{equation}

	What follows is the version of the Hahn-Banach theorem that underlies almost all theorems in this article and its companion article, \cite{feinberg-lavineEntropy2}.
	
\begin{theorem}{\emph{\textbf{(Hahn-Banach)}}}\label{thm:HahnBanach} Let $V$ be a Hausdorff locally convex topological vector space, and let $A$ and $B$ be non-empty disjoint closed convex subsets of $V$, with $B$ compact. There is a  continuous linear function $f: V \to\, \mathbb{R}$ and a number $\gamma \in \mathbb{R}$ such that
\begin{equation}
f(a)\  <\  \gamma,\  \forall\  a \in A\nonumber
\end{equation}
and
\begin{equation}
f(b)\  >\  \gamma,\  \forall\  b \in B\nonumber.
\end{equation}
In particular, if $A$ is a cone, then
\begin{equation}
f(a)\  \leq\  0,\  \forall\  a \in A\nonumber
\end{equation}
and
\begin{equation}
f(b)\  >\  0,\  \forall\  b \in B\nonumber.
\end{equation}
\end{theorem}

\begin{rem} For proofs of this version of the Hahn-Banach theorem see Theorem 21.12 in \cite{choquet1969lectures}, Theorem 1.7 in \cite{brezis2011functional}, or Corollary 14.4 in \cite{kelley1963linear}. The last sentence of Theorem \ref{thm:HahnBanach} is not usually stated explicitly, but it is an easy consequence of the preceding one.
\end{rem}
\medskip

	We are now in a position to prove Theorem \ref{thm:ExistTempEnt}, the central theorem of this article.

\begin{proof}[Proof of Theorem \ref{thm:ExistTempEnt}] To prove that (i) implies (ii) we first note for the Kelvin-Planck theory \theory\  that, in the Hausdorff locally convex topological vector space \VSigma, the closed convex cone $\hat{\scrP}$ is disjoint from the convex compact set $(0,\MSigmaPlOne)$. From the Hahn-Banach theorem, then, there is a continuous linear function $f: \VSigma \to \mathbb{R}$ such that
\begin{equation}\label{eq:FirstHBIneq}
f\,\process\,\leq \,0, \quad \forall\,\process\,\in\,\hat{\scrP}
\end{equation}
and
\begin{equation}\label{eq:SecondHBIneq}
f(0,\scrw)\, > \, 0,  \quad \forall\,(0,\scrw)\,\in\,(0,\MSigmaPlOne).
\end{equation}
Moreover, there are functions $\eta\,(\cdot)$ and $\beta\,(\cdot)$ in $\textnormal{C}(\Sigma,\mathbb{R})$ such that $f(\cdot,\cdot)$ has the representation\footnote{See, for example, \S 3.14 in \cite{rudin_functional_1991}. Although for every continuous linear functional $g$ on \MSigma\  there is a unique continuous function $\varphi \in \textnormal{C}(\Sigma,\mathbb{R})$ such that $g(\mu) = \int_{\Sigma}\varphi\, d \mu,\  \forall \mu \in \MSigma$, the situation for \MSigmaZ, with the topology given earlier, is a little different. In that case, the representing function $\varphi$ is unique only up to an additive constant.}
\begin{equation}\label{eq:HBRep}
f(\scrv,\scrw) = \int_{\Sigma}(-\eta)\,d\scrv + \int_{\Sigma}\beta\,d\scrw,\quad \forall (\scrv,\scrw) \in \VSigma.
\end{equation} 
Note that for each $\sigma \in \Sigma$ the Dirac measure $\delta_{\sigma}$ is a member of \MSigmaPlOne. From this, \eqref{eq:SecondHBIneq}, and \eqref{eq:HBRep} it follows that $\beta(\cdot)$ takes strictly positive values. Letting $T(\cdot) = 1/\beta(\cdot)$, we get \eqref{eq:MainThmCDIneq} as a consequence of \eqref{eq:FirstHBIneq} and \eqref{eq:HBRep}. This completes proof that (i) implies (ii).

	To prove that (ii) implies (i) we first observe that if the inequality \eqref{eq:MainThmCDIneq} is satisfied for a particular $\process \in \scrP$, then the inequality	
is also satisfied by $\alpha\process$ for every non-negative number $\alpha$. For this reason, (ii) implies that the inequality
\begin{equation}\label{eq:MainThmIneq}
\int_{\Sigma}\eta\: d(\scrv)\  \geq\  \int_{\Sigma}\frac{d\scrw}{T}
\end{equation}
is satisfied for all $(\scrv,\scrw)$ in $\Cone(\scrP)$ and therefore  for all  $(\scrv,\scrw)$ in  $\hat{\scrP} := \textrm{cl}\,(\Cone(\scrP)]$. To show that \theory\  is a Kelvin-Planck theory we must show that $\hat{\scrP}$ can contain no member of the form $(0,\scrw)$, where \scrw\ is a nonzero member of \MSigmaPl. Because $T(\cdot)$ is positive-valued, such an element could not satisfy \eqref{eq:MainThmIneq}. This completes the proof of Theorem \ref{thm:ExistTempEnt}.
\end{proof}
%\smallskip

\begin{rem}[\emph{Interpretation of (ii)}] 
In Theorem \ref{thm:ExistTempEnt} (ii) we will, of course, regard \eqref{eq:MainThmCDIneq} to be an expression of the Clausius-Duhem inequality, with  $\eta(\cdot)$ and $T(\cdot)$ playing the roles of specific-entropy (entropy per mass) and thermodynamic temperature functions of state  that assign to each $\sigma \in \Sigma$   a specific-entropy $\eta(\sigma)$  and a value $T(\sigma)$ of the thermodynamic temperature. 

If, for a physical process, $\scrm_i$ and $\scrm_f$ are the initial and final conditions of the body suffering the process then, with $\deltam =\scrm_f - \scrm_i$, we have 
\begin{equation}
\label{eq:EntropyDifferenceBody}
\int_{\Sigma}\eta\: d(\deltam) = \int_{\Sigma}\eta\: d \scrm_f -\int_{\Sigma}\eta\: d \scrm_i.
\end{equation}
In view of \eqref{eq:EntropyDifferenceBody} we can interpret the integral on the left side of \eqref{eq:MainThmCDIneq} to be the difference in the entropy of the body suffering the process between the end of the process and its beginning. 

	In this sense, Theorem \ref{thm:ExistTempEnt} tells us that for any Kelvin-Planck theory, there is a notion of the entropy of a \emph{body} (along with a thermodynamic temperature scale) that aligns with the Gibbs version \eqref{eq:CDInGibbs} of the Clausius-Duhem inequality with which we began. \emph{Note, however, that Theorem \ref{thm:ExistTempEnt} does much more, for it provides, in the spirit of modern classical physics, a \emph{local} notion of specific entropy (entropy per mass), as a function of \emph{local} state within a body.}

	If a particular process \process\  derives from the data specified in the example of  Section \ref{subsubsec:ContinuumMechEx}, the inequality in (ii) can be pulled back to a more traditional description of the Clausius-Duhem inequality, in effect an elaboration of the Gibbs version \eqref{eq:CDInGibbs} suited to modern continuum physics:
\begin{equation}
\int_{\scrB}\eta\,(\hat{\sigma}_f(X))\,d\mu(X) - \int_{\scrB}\eta\,(\hat{\sigma}_i(X))\,d\mu(X) \geq \int_{\scrB \times \scrI}\frac{d\,h(X,t)}{T(\hat{\sigma}(X,t))}.
\end{equation}	
Connections of entropy (with existence derived via  \cite{feinberg1986foundations}) to the theory of partial differential equations (in particular the canonical equations of continuum physics) are discussed by L. C. Evans in \cite{evans2004entropy}.
\end{rem}
\bigskip

In preparation for our concluding remarks and for the companion article \cite{feinberg-lavineEntropy2}, we record the following definition:

\begin{definition}[\emph{Entropy, Thermodynamic Temperature}] 
\label{def:CDpair}
Let \theory\ be a Kelvin-Planck theory. An element $(\eta,T)$ of $\textnormal{C}(\Sigma,\mathbb{R}) \times \textnormal{C}(\Sigma,\mathbb{R}_+)$ that satisfies \eqref{eq:MainThmCDIneq} is a \textbf{Clausius-Duhem pair} for the theory. A function $T \in \textnormal{C}(\Sigma,\mathbb{R}_+)$ is a \textbf{Clausius-Duhem temperature scale} for the theory if there exists $\eta \in \textnormal{C}(\Sigma,\mathbb{R})$ such that $(\eta,T)$ is a Clausius-Duhem pair. In that case, $\eta(\cdot)$ is a \textbf{specific-entropy function} for the theory (corresponding to the Clausius-Duhem temperature scale $T(\cdot))$. 
\smallskip
\end{definition} 

\begin{rem}[\emph{Differentiability of the specific-entropy function and the thermodynamic temperature scale}] In applications of the Clausius-Duhem inequality, differentiability of the entropy and temperature with respect to state descriptors often plays a role. Here we focused solely on continuity of these functions. When, for a thermodynamic theory \theory, the state space $\Sigma$ is such that differentiability of real-valued functions on $\Sigma$ has meaning, Theorem \ref{thm:ExistTempEnt} remains true with $\textnormal{C}(\Sigma,\mathbb{R})$ replaced by $\textnormal{C}^{\,k}(\Sigma,\mathbb{R})$, so long as the same replacement is made in the definition of the topology on \MSigma, given in footnote \ref{ft:WeakStarDef}. That revised topology, which is coarser than the weak-star topology, exerts itself  in the definition of $\hat{\scrP} := \textrm{cl}\,(\Cone(\scrP))$. This is discussed more fully, but in a narrower context, in Remark 10.2 of \cite{feinberg1983thermodynamics}. Similar considerations apply to the theorems of the companion article \cite{feinberg-lavineEntropy2}. 
\end{rem}

\medskip
\section{Concluding Remarks} In any thermodynamical theory that complies with the Kelvin-Planck Second Law, as expressed by \eqref{eq:KPSecLawDef}, Theorem \ref{thm:ExistTempEnt} asserts that there are  invariably  specific-entropy and thermodynamic-temperature functions (of the \emph{local} material state) that together satisfy the Clausius-Duhem condition \eqref{eq:MainThmCDIneq}.  Moreover, the two conditions are \emph{equivalent}, so any theory for which there is a Clausius-Duhem entropy-temperature pair \emph{must} comply with the  form   of the Kelvin-Planck Second Law given by \eqref{eq:KPSecLawDef}. 

	Again, the proof that (i) implies (ii) is immediate. It relies only on the Hahn-Banach Theorem and functional analysis infrastructure unavailable to the brilliant founders of classical thermodynamics. It is worth emphasizing again that, with respect to the \emph{existence} of Clausius-Duhem entropy-temperature pairs, there is no reliance on reversible processes or notions of thermodynamic equilibrium. There is no requirement that the set of processes contain certain ones of a specified kind. To some extent this will change in the companion article \cite{feinberg-lavineEntropy2}, where we consider properties (including uniqueness) of specific-entropy and thermodynamic-temperature functions of state, in particular the relation of those properties to the supply of processes.

\renewcommand{\theequation}{A.\arabic{equation}}
  % redefine the command that creates the equation no.
 \renewcommand{\thetheorem}{A.\arabic{theorem}}
  \setcounter{theorem}{0}  % reset counter 
 % \section*{APPENDIX}  % use *-form to suppress numbering

\section*{Appendix: The Convexity of  $\hat{\scrP} $}
\addcontentsline{toc}{section}{Appendix: The Convexity of  $\hat{\scrP}$}
\label{app:ConvClConeC}

In the main body of this article, the set  $\scrP\ \subset \VSigma$ carried information about the totality of outcomes admitted by processes within a particular thermodynamical theory.  In this appendix we will argue that, in natural theories, \scrP\ can be expected to have a special structure. In particular,  we will provide support for the presumption in the main text that $\hat{\scrP} : =\textnormal{cl}\,[\textnormal{Cone}\,(\scrP)]$ is convex.

Recall that an element of \scrP, say $\scrp = (\Delta \scrm, \scrq)$,  provides information about the \emph{overall} result of a particular process, with $\Delta \scrm := \scrm(t^f)  -  \scrm(t^i)$ giving the overall difference, from the initial time to the final time, in the condition of the body suffering the process and with \scrq\  giving the process's overall heating measure.  Although the emphasis has been on the overall outcome, a physical process nevertheless evolves over time in the instants between its inception and completion.

With this in mind, we take for granted that each process $\scrp \in \scrP$ can be associated with a physical history experienced by a particular body.\footnote{There might be several histories that different bodies can experience which nevertheless result in the same overall record carried by \scrp. These different histories might have different durations.  Our presumption here is that for each \scrp\ there is at least one such history.} In particular, with \scrp\  we can associate a closed time interval $[t^i_{\scrp},t^f_{\scrp}]$, where $t^i_{\scrp}$ is the initial time at which the process begins and $t^f_{\scrp}$ is the final time at which it ends. The \emph{duration} of the process history is the positive number $t^f_{\scrp} - t^i_{\scrp}$. 

 Moreover, we assume that we can associate, at every instant in $[t_{\scrp}^i,t_{\scrp}^f]$, a specification of  the  difference between the body's current condition and its condition at the process's inception, and that we can also associate a specification of the current (cumulative) heating measure.  That is, with process \mbox{$\scrp\ = (\Delta \scrm, \scrq)$}  we assume that there is a continuous \emph{process history}\footnote{Here is is understood that the functions $\Delta\bar\scrm(\cdot)$ and $\bar\scrq(\cdot)$ are particular to the process history under consideration and that $\Delta\bar\scrm(\tau) = \bar\scrm(\tau) - \bar\scrm(t^i_{\scrp})$, where $\bar\scrm(\tau)$ gives the condition at time $\tau$ of the body suffering the process.} 
 
\begin{equation}
\tau \in [t^i_{\scrp},t^f_{\scrp}] \to (\Delta\bar{\scrm}(\tau),\bar{\scrq}(\tau)) \in \VSigma
\end{equation}

\medskip\noindent
such that   $\bar{\scrq}(t^i_{\scrp}) =0$, $\bar{\scrq}(t^f_{\scrp}) = \scrq$, $\Delta\bar{\scrm}(t^i_{\scrp}) = 0$, and  $\Delta\bar{\scrm}(t^f_{\scrp}) = \Delta\scrm$.

\smallskip

 \begin{rem} \label{subproc} Given the process history described above, we take for granted that, for each closed time interval contained within $[t^i_{\scrp},t^f_{\scrp}]$, there is another member of \scrP, say  $\scrp^* = (\Delta \scrm^*, \scrq^*)$,  corresponding to the restriction of the given process history to that smaller time interval. That is, if $[t^i_{\scrp^*},t^f_{\scrp^*}]$ is the smaller time interval, then 
\begin{equation}
\Delta \scrm^*  = \Delta \bar{\scrm}(t^f_{\scrp^*}) - \Delta \bar{\scrm}(t^i_{\scrp^*})\quad \textrm{and}\quad \scrq^* = \bar{\scrq}(t^f_{\scrp^*}) - \bar{\scrq}(t^i_{\scrp^*}).
\end{equation}
 \end{rem}
 
 \bigskip%\medskip

	With this as background, what follows is a brief list of properties we assume to be possessed by the set of process histories associated with \scrP. Each property will be accompanied by a rationale. Taken together, these properties will shed light on the geometric structure of \scrP.
	
\begin{property} If $\scrp_1$ and $\scrp_2$ are members of \scrP\  associated with process histories of identical duration, then $\scrp_1 + \scrp_2$ is also a member of \scrP\ having associated with it a history of that same duration. \label{Prop1}
\end{property}

\begin{rationale} If the two processes $\scrp_1 = (\Delta \scrm_1, \scrq_1)$ and $\scrp_2 = (\Delta \scrm_2, \scrq_2)$ are experienced by bodies $\scrB_1$\ and $\scrB_2$ then those same processes can be run simultaneously with copies of $\scrB_1$\ and $\scrB_2$ at remote locations (or, more generally, thermally insulated from each other). The union of the bodies is again a body. The process experienced by the union will have $\Delta \scrm_1 + \Delta \scrm_2$ as the body's change of condition and $\scrq_1 + \scrq_2$ as its heating measure. Thus, $\scrp_1 + \scrp_2$ is a member of \scrP.   
\end{rationale}

\begin{rem} \label{remNscrp} If \scrp\  is a member of \scrP\ there is, by supposition, a process history associated with it. It is a consequence of Property \ref{Prop1} (and its rationale) that, for any integer $n$, $n\scrp$ is also a member of \scrP.
\end{rem}

\begin{property} If $\scrp =(\Delta \scrm, \scrq) \in \scrP$ has associated with it a process history of duration $d$ then, for any integer $N$, \scrp\ also has associated with it a history of duration $d/N$.\label{Prop2}
\end{property}

\begin{rationale}
  The time interval for the given process history can be regarded to be the union of N sequential (closed) time intervals, each of duration $d/N$. With each such smaller interval we can associate, as in Remark \ref{subproc}, a sub-process history. Using $N$ copies of the original body suffering the process, we can execute those $N$ sub-process histories simultaneously, as in the rationale for Property \ref{Prop1}. The union of the $N$ body-copies is again a body, this one suffering a process of duration $d/N$. By virtue of Property \ref{Prop1} the overall change of condition will again be $\Delta \scrm$ and the overall heating measure will again be \scrq.
\end{rationale}

\begin{property}  Suppose that two members of \scrP, say $\scrp_1$ and $\scrp_2$, are associated with histories of durations $d_1$ and $d_2$. If $d_1/d_2$ is rational, then $\scrp_1+\scrp_2$ is also a member of \scrP. \label{Prop3}
\end{property}

\begin{rationale} This is just a consequence of Properties \ref{Prop1} and \ref{Prop2}: Suppose that $d_1/d_2 = N_1/N_2$, where $N_1$ and $N_2$ are integers. Then, from Property \ref{Prop2}, $\scrp_1$ and $\scrp_2$ can be associated with two histories of identical duration, $d_1/N_1 = d_2/N_2$. From Property \ref{Prop1} it follows that $\scrp_1+\scrp_2$ is a member of \scrP.
\end{rationale}
\medskip

We will not assume that we can always associate with $\scrp \in \scrP$ a process history of rational duration. Nevertheless, we will assume that there is invariably a nearby element of \scrP\  that can be associated with a rational-duration process history. This is made precise in the following way:

\begin{property}\label{Prop4} If \scrp\ is an element of \scrP\ and $\scrO \subset \VSigma$ is any open neighborhood of \scrp, there is in \scrO\ an element of \scrP\  that can be associated with a process history of rational duration.
\end{property}

\begin{rationale} This is just a consequence of the  natural assumption that \scrp\  can be associated with a process history that is continuous in time. In fact, for Property \ref{Prop4} to obtain it is sufficient that \scrp\ can be associated with a process history that is merely continuous at its final time.
\end{rationale}

\bigskip

	In the following proposition we assume that, for the thermodynamic theory $(\Sigma, \scrP)$ under consideration, the set \scrP\  and the set of histories associated with its elements possess Properties \ref{Prop1}-\ref{Prop4}. 
	
\begin{proposition}\label{PropConvexity}  $\hat{\scrP} : =\textnormal{cl}\,[\textnormal{Cone}\,(\scrP)]$ is a convex subset of \VSigma.
\end{proposition}

\begin{proof} We need to show that if $x_1$ and $x_2$ are members of $\hat{\scrP}$ and $\alpha$ is a number between $0$ and $1$, then $\alpha x_1 + (1-\alpha)x_2$ is also a member of $\hat{\scrP}$. It is not difficult to show that, in a topological vector space, the closure of a cone is again a cone, whereupon $\hat{\scrP}$ is a cone. In particular $\alpha x_1$ and  $(1-\alpha)x_2$ are members of $\hat{\scrP}$.    Therefore, to establish that $\hat{\scrP}$ is convex, it is enough to prove the following lemma:

\begin{lemma} If $v_1$ and $v_2$ are members of $\hat{\scrP}$  then so is $v_1 + v_2$.
\end{lemma}

\begin{proof} Our aim is to show that if, in the topological vector space $\VSigma$, \scrO\  is an open neighborhood of $v_1 + v_2$, then \scrO\ contains an element of \coneP.
 
 Because vector addition $\VSigma \times \VSigma \to \VSigma$  is continuous, there are open sets $\scrO_{\,1}$ and $\scrO_{\,2}$ in \VSigma\  containing $v_1$ and $v_2$ respectively such that the set
\begin{equation}
\scrO_{\,1} + \scrO_{\,2} := \{w_1 + w_2 \in \VSigma\,: w_1 \in  \scrO_{\,1}, w_2 \in  \scrO_{\,2}\}
\end{equation} 
is contained in \scrO.  

By supposition, $v_1$ is a member of \textnormal{cl\,[\,cone\,(\scrP)\,]}. Because $\scrO_{\,1}$ is an open neighborhood of $v_1$, there must be a member of $\Cone(\scrP)$ in  $\scrO_{\,1}$. That is,  $\scrO_{\,1}$ must contain a member of the form $\alpha_1\scrp_1$, with $\alpha_1$ a positive number and $\scrp_1$\  a member of \scrP. Because, in the topological vector space \VSigma, scalar multiplication $\mathbb{R} \times \VSigma \to \VSigma$  is continuous, there is an open interval $I_1 \subset \RP$ containing $\alpha_1$ and an open neighborhood $\hat{\scrO}_1$ of $\scrp_1$ such that the set 

\begin{equation}
I_1 \cdot \hat{\scrO}_1 := \{\theta_1 w_1 \in \VSigma : \   \theta_1  \in I_1,  w_1 \in \hat{\scrO}_{\,1}, \}
\end{equation} 
is contained in $\scrO_1$. 

In particular, there is a \emph{rational} number $\alpha^*_1 \in I_1$ and, from Property \ref{Prop4}, an element $\scrp^*_1 \in \hat{\scrO_1}$  associated with a process history of \emph{rational} duration such that $\alpha^*_1\scrp^*_1$ is a member of $\scrO_1$. Similarly, $\scrO_2$ contains a member of the form $\alpha^*_2\scrp^*_2$, where $\alpha^*_2$ is rational and $\scrp^*_2 \in \scrP$ has associated with it a process history of rational duration. Because $\scrO_{\,1} + \scrO_{\,2}$ is contained in \scrO, we have the inclusion 
\begin{equation}
\alpha^*_1\scrp^*_1 + \alpha^*_2\scrp^*_2 \in \scrO.
\end{equation}
It remains to be shown that $\alpha^*_1\scrp^*_1 + \alpha^*_2\scrp^*_2$  is  a member of $\Cone(\scrP)$.

	Let $n_1, m_1, n_2, m_2$ be integers such that $\alpha_1^* = n_1/ m_1$ and $\alpha_2^* = n_2/ m_2$. Thus, we have the inclusion
\begin{equation}
\frac{1}{m_1m_2}(n_1m_2\scrp^*_1 + n_2m_1\scrp^*_2) \in \scrO.
\end{equation}
From Remark \ref{remNscrp} it follows that $n_1m_2\scrp^*_1$ and $n_2m_1\scrp^*_2$ are 
members of \scrP\ having associated with them individual process histories of rational durations (identical to those associated with $\scrp^*_1$ 
and $\scrp^*_2$, respectively). From Property \ref{Prop3}, then, their sum $\scrp^{**} := n_1m_2\scrp^*_1 + n_2m_1\scrp^*_2$ is a member of  \scrP, so we have the inclusion 
\begin{equation}
\frac{1}{m_1m_2}\scrp^{**} \in \scrO.
\end{equation}	
Thus, there is  a member of \coneP\  that lies in \scrO. This is what we wanted to prove.
\end{proof} 
This completes the proof of Proposition \ref{PropConvexity}.
\end{proof}

\bibliographystyle{spmpsci}
\bibliography{MonoLibrary-Thermo}

\end{document}